\documentclass[11pt,reqno,letterpaper]{amsart}
\usepackage{booktabs}
\usepackage{color}
\usepackage[colorlinks=true, allcolors=blue, hypertexnames=false]{hyperref}
\usepackage{amsmath, amssymb, amsthm}
\usepackage{mathrsfs}
\usepackage{mathtools}
\usepackage[noabbrev,capitalize,nameinlink]{cleveref}
\usepackage{fullpage}
\usepackage[noadjust]{cite}
\usepackage{graphics}
\usepackage{pifont}
\usepackage{tikz}
\usepackage{bbm}
\usepackage{thm-restate}
\usepackage[T1]{fontenc}

\usetikzlibrary{arrows.meta}

\usepackage{environ}
\usepackage{framed}
\usepackage{url}
\crefname{algocf}{Algorithm}{Algorithms}
\Crefname{algocf}{Algorithm}{Algorithms}
\usepackage[linesnumbered,ruled,vlined]{algorithm2e}
\usepackage[noend]{algpseudocode}
\usepackage[labelfont=bf]{caption}
\usepackage{cite}
\usepackage{framed}
\usepackage[framemethod=tikz]{mdframed}
\usepackage{appendix}
\usepackage{graphicx}
\usepackage[textsize=tiny]{todonotes}
\usepackage{tcolorbox}
\usepackage{enumerate}
\allowdisplaybreaks[1]
\usepackage{enumerate}
\usepackage{stmaryrd}
\usepackage[margin=1in]{geometry}
\usetikzlibrary{decorations.pathreplacing}
\allowdisplaybreaks

\newtheorem{theorem}{Theorem}[section]
\newtheorem{proposition}[theorem]{Proposition}
\newtheorem{lemma}[theorem]{Lemma}

\newtheorem{remark}[theorem]{Remark}

\theoremstyle{definition}

\numberwithin{equation}{section}

\newcommand{\cB}{\mathcal{B}}
\newcommand{\E}{\mathbb{E}}
\renewcommand{\P}{\mathbb{P}}
\newcommand{\R}{\mathbb{R}}

\newcommand{\eps}{\varepsilon}
\newcommand{\mc}{\mathcal}
\newcommand{\on}{\operatorname}

\newcommand{\paren}[1]{\left( #1 \right)}
\newcommand{\floor}[1]{\left\lfloor #1 \right\rfloor}

\newcommand{\mb}{\mathbb}
\newcommand{\mbf}{\mathbf}
\DeclareMathOperator{\LIS}{LIS}

\title{The online monotone array completion problem}

\author{Vishesh Jain}
\address{Department of Mathematics, Statistics, and Computer Science, University of Illinois Chicago, Chicago, IL 60607, USA}
\email{visheshj@uic.edu}

\author{Dylan King}
\address{Department of Mathematics, California Institute of Technology, Pasadena, CA 91106}
\email{dking@caltech.edu}

\author{Clayton Mizgerd}
\address{Department of Mathematics, Statistics, and Computer Science, University of Illinois Chicago, Chicago, IL 60607, USA}
\email{cmizge2@uic.edu}

\begin{document}

\begin{abstract}
Consider the following online filling game.  An array of length \(n\) is initially
empty.  At each time step one observes an independent sample from
\(\mathrm{Unif}[0,1]\) and must either discard it or place it irrevocably into an empty
position of the array, while preserving the constraint that the occupied entries are
non-decreasing from left to right.  Among all possible strategies, what is the
optimal expected time required to fill the array?

Let \(v_n\) denote this optimal expected completion time.  Our main result determines
\(v_n\) up to lower-order terms:
\[
v_n=\left(\frac12+o(1)\right)n\log n.
\]
More precisely, no strategy, even if randomized and adaptive, can have expected
completion time below \(\left(\frac12-o(1)\right)n\log n\), while we provide an explicit
deterministic strategy whose expected completion time is at most
\(\left(\frac12+o(1)\right)n\log n\).  For comparison, the natural coupon-collector
strategy, which partitions \([0,1]\) into \(n\) equal intervals and reserves one array
position for each interval, has expected completion time \((1+o(1))n\log n\).

We also consider a with-replacement version of the game, in which previously placed
entries may be overwritten.
For this variant, we give
a deterministic strategy with expected completion time \(O(n\sqrt{\log n})\), thereby
establishing a separation between the two models.
\end{abstract}

\maketitle

\section{Introduction}

We study an online array-filling game with a monotonicity constraint.  An array
of length \(n\) is initially empty.  At each step, a new independent sample from
\(\mathrm{Unif}[0,1]\) is revealed.  The player must either discard the sample forever
or place it into one of the empty positions of the array, subject to the constraint
that the occupied entries remain non-decreasing from left to right.  The goal is to
fill the array as quickly as possible.

Formally, a \emph{partial array state} is a vector
\[
a=(a_1,\dots,a_n)\in \paren{[0,1]\cup\{\ast\}}^n
\]
such that whenever \(1\le i<j\le n\) and \(a_i,a_j\in[0,1]\), one has
\(a_i\le a_j\).  The symbol \(\ast\) denotes an empty coordinate.  Starting from the
all-empty state, at each integer time \(t\geq 1\) the player observes a fresh sample
\(X_t\sim \mathrm{Unif}[0,1]\), independent of the past.  The player then either
discards \(X_t\), or places it into an empty coordinate in a way that preserves the
partial-array condition.  The game ends once all \(n\) coordinates have been filled.

For a strategy \(S\), let \(\tau_S\) denote its completion time, and set
\[
T_S(n):=\E \tau_S.
\]
We write
\[
v_n:=\inf_S T_S(n),
\]
where the infimum ranges over all possibly randomized adaptive strategies in this
game.

Let us first record two elementary bounds on \(v_n\).  The lower bound
\[
v_n\ge n
\]
is immediate, since each sample can fill at most one previously empty coordinate.
For an upper bound, consider the following coupon-collector strategy.  Partition
\([0,1]\) into the intervals
\[
A_i=\left[\frac{i-1}{n},\frac{i}{n}\right),\qquad 1\le i\le n,
\]
with the usual harmless convention at the right endpoint.  The strategy places a
sample in position \(i\) only if it lies in \(A_i\) and position \(i\) is still empty;
all other samples are discarded.  The array is filled exactly when every interval
\(A_i\) has been hit at least once, so (by the well-known coupon-collector calculation) this strategy has expected completion time
\[
nH_n=(1+o(1))n\log n,
\]
where~$H_n = \sum_{i=1}^ni^{-1}$ is the~$n$-th harmonic number.
Thus
\[
n\le v_n\le nH_n=(1+o(1))n\log n.
\]

Since a general strategy may be adaptive and may use auxiliary randomness independent
of the sample sequence, it is not \emph{a priori} clear whether the logarithmic factor
in the coupon-collector strategy is necessary.  Our main result shows that it is:
every strategy requires order \(n\log n\) samples.  On the other hand, the
coupon-collector strategy is not optimal in its leading constant.  There are
strategies which complete the array twice as fast, to first order, and this improvement
is best possible.

\begin{theorem}\label{thm:main}
With notation as above,
\[
v_n=\left(\frac12+o(1)\right)n\log n.
\]
More precisely, for every sequence of possibly randomized adaptive strategies
\(S^{(n)}\),
\[
T_{S^{(n)}}(n)\ge \left(\frac12-o(1)\right)n\log n,
\]
while there is an explicit deterministic sequence of strategies \(\widehat S^{(n)}\)
satisfying
\[
T_{\widehat S^{(n)}}(n)\le \left(\frac12+o(1)\right)n\log n.
\]
\end{theorem}

\begin{remark}
The proofs give corresponding polynomial lower and upper tail bounds; see
\cref{prop:hp-lower,prop:upper-explicit}.
\end{remark}

\subsection{The with-replacement variant} 
We also consider a natural with-replacement version of the game.  The array again
begins empty and must remain monotone at all times, but the player may now overwrite
previously placed entries.  More precisely, after seeing \(X_t\), the player may either
discard it or place it in any coordinate \(j\), whether or not that coordinate is already
occupied, provided that the resulting partial array state is monotone.

Let \(v_n^{\mathrm{rep}}\) denote the analogue of \(v_n\) for this with-replacement
game.  Since each coordinate must eventually be occupied, we still have the trivial
lower bound
\[
v_n^{\mathrm{rep}}\ge n.
\]
One might expect the with-replacement version to have the same order of completion
time as the original game.  Our next result shows, perhaps surprisingly, that this is not the case.

\begin{theorem}\label{thm:replacement-upper}
For the with-replacement game,
\[
v_n^{\mathrm{rep}}=O(n\sqrt{\log n}).
\]
More precisely, for each \(n\), there is a deterministic with-replacement strategy
\(S_{\mathrm{rep}}^{(n)}\) such that
\[
\E[\tau_{S_{\mathrm{rep}}^{(n)}}]=O(n\sqrt{\log n}).
\]
Moreover, for every fixed \(D>0\), there is a constant \(C_D\) such that
\[
\P\left(\tau_{S_{\mathrm{rep}}^{(n)}}>C_D n\sqrt{\log n}\right)\le n^{-D}
\]
for all sufficiently large \(n\).
\end{theorem}

The above theorem shows that the with-replacement and without-replacement models are separated. Unlike in the without-replacement case, we do not know the correct order of $v_n^{\on{rep}}$. In particular, we leave both improving the upper bound beyond $O(n\sqrt {\log n})$ and the lower bound beyond $\Omega(n)$ as open problems. 

\subsection{Related work}
The online monotone array completion problem belongs to a broad family of online selection problems for random sequences. The closest classical line of work is the online selection of increasing subsequences, initiated by Samuels and Steele in 1981 \cite{SamuelsSteele81} and subsequently developed by various researchers over the next four decades; see, for instance, \cite{SamuelsSteele81,Gnedin99,BrussDelbaen01,BrussDelbaen04,ArlottoNguyenSteele15}.  In that problem, i.i.d.~$\on{Unif}[0,1]$ observations are revealed sequentially and each observation must be accepted or rejected immediately, with accepted observations required to form an increasing subsequence in their order of arrival. In particular, a ``dual version'', in which one minimizes the expected time needed to select an increasing subsequence of a prescribed length, was studied by Arlotto, Mossel and Steele \cite{ArlottoMosselSteele16}. 

Our problem differs from these subsequence-selection problems in one important respect.  A sample that is accepted in our setting need not be appended to the right end of the selected sequence.  Instead, it may be inserted into any still-empty position in the array provided the monotonicity constraint is not violated.  This additional spatial freedom completely changes the nature of the obstruction.  The difficulty is no longer building a long increasing sequence in the temporal order of arrivals, but rather, to fill all remaining gaps in a partially ordered array before the admissible value intervals around those gaps become too small.

There is also a separate online-sorting literature in which arriving real numbers
must be placed irrevocably into an array, and the objective is to minimize a final
cost such as
\[
\sum_{i=2}^n |a_i-a_{i-1}|.
\]
This line was initiated by Aamand, Abrahamsen, Beretta, and Kleist
\cite{AamandAbrahamsenBerettaKleist23}, and has since been developed in randomized,
stochastic, and high-dimensional variants
\cite{AbrahamsenBerceaBerettaKlausenKozma24,Hu26}.  Although this model also
involves online placement into an array, its objective and constraints are different
from ours: in our problem samples may be discarded, monotonicity must be preserved
throughout, and the quantity of interest is the completion time.

\subsection{Overview of the proof} The main work is in the lower-bound argument in \cref{sec:lower}.  The 
difficulty is that the strategy is completely arbitrary: it may be randomized,
history-dependent, and may place an accepted sample in any legal empty position subject to the monotonicity constraint. In particular, the partial array state need not have any simple or predictable structure. A lower bound must therefore identify some quantity which obstructs completion
uniformly over all possible states that any strategy can create. 

We accomplish this by tracking a potential on the current gaps.  At time \(t\), the empty
coordinates form maximal consecutive blocks.  If \(B\) is such a block, let
\(c(B)\) be its number of empty coordinates and let \(I(B)\) be the interval of
values which can still be legally placed in \(B\).  For a parameter
\(a\in(0,n+1)\), we define
\[
q_a(B):=(c(B)+1-a|I(B)|)_+,
\qquad
Q_t:=\sum_{B\in\cB_t}q_a(B),
\]
where the sum is taken over~$\cB_t$, the collection of blocks at time~$t$. Thus \(Q_t\) measures the extent to which the remaining empty blocks have too
much capacity relative to their available value intervals.  Initially \(Q_0\) is
large, while at completion \(Q_\tau=0\), so every strategy must eventually drive
this potential to zero.

The key point is that this potential cannot be decreased too quickly, even by a
fully adaptive strategy. Once we condition on the current state, every accepted
sample has exactly one of three effects on the block structure.  It may be an
\emph{interior move}, which places the sample inside a block and splits that
block into two nonempty child blocks.  It may be an \emph{edge move}, which
places the sample in the leftmost or rightmost empty position of a block of
capacity at least two.  Or it may be a \emph{singleton completion}, which fills
and deletes a block of capacity one. No other kind of update is possible because previously placed values are fixed and hence
serve as fixed boundaries for the adjacent blocks.  We emphasize that this is the point where the lower-bound argument uses the without-replacement assumption in a crucial way; in the with-replacement model, a later sample could move an already placed
boundary value, changing the admissible intervals of neighboring gaps without
filling a new coordinate.

It is easily seen (\cref{lem:easy_properties}) that an interior move cannot decrease $Q_t$. Hence, the only moves which can
actually reduce $Q_t$ are edge moves, which can decrease $Q_t$ by at most $1$, or singleton completions, which can decrease $Q_t$ by at most $2$. However, for the next sample to trigger an edge move or singleton completion which actually decreases $Q_t$, it must lie in a set of small measure.  Roughly, for an edge move to reduce the
potential in a block \(B\), the incoming sample must lie very close to one of the
two endpoints of \(I(B)\), while for a singleton completion to reduce the potential,
the admissible interval of that singleton must itself be short.  Formalizing this intuition, we show in the crucial \cref{lem:expected_positive_decrement} that no matter what the strategy is and how it 
has reached the present configuration,
\[
\E\bigl[(Q_t-Q_{t+1})_+\mid \mc F_t\bigr]
\le
\frac{2Q_t}{a},
\]
where~$\mc F_t$ is the filtration containing the entire present state of the partial array.
It remains to convert this one-step estimate into a lower bound on the completion time. Since $Q_t$ can decrease at most at a proportional rate, it is natural to measure progress on a logarithmic scale. Accordingly, we consider 
\[\Phi(Q_t) = \log(1+Q_t/2) = \int_0^{Q_t} \frac{ds}{s+2};\]
the shift by $2$ in $s+2$ matches the fact that $Q_t$ can decrease by at most $2$ in one step. With this choice, the
expected decrease of \(\Phi(Q_t)\) in any step is \(O(1/a)\), uniformly
over the current state and over the strategy.  When \(a = n - \Theta(n/\log n)\), initially
\(\Phi(Q_0)\) is of order \(\log n\), while at completion \(\Phi(Q_\tau)=0\), and therefore the strategy needs \(\Omega(a\log n)\) samples.  Optimizing \(a\) and keeping track of the constants gives the
\((1/2-o(1))n\log n\) lower bound in \cref{thm:main}.

The upper bound in \cref{sec:upper} is a nontrivial block version of the
coupon-collector strategy.  We divide the array into blocks, assign each block
its own interval of values, and try to complete all blocks in parallel.  The key
ingredient is the strategy used inside each block, which can be viewed as
reverse-engineering the obstruction from the lower bound.  Recall that the
lower bound identifies edge moves and singleton completions as the only moves
which can decrease the potential.  Motivated by this, the strategy inside each
block fills only from the two ends of the remaining empty run, so that every
accepted move is an edge move, except for the final singleton completion.  The
main invariant, \cref{lem:feasible-length}, says that each unfinished block has
an acceptance region of length essentially \(2/n\).  Thus the time to complete
the array is governed by the slowest among about \(n/b\) coupon-collector-type
blocks (where \(b\) is a parameter we will eventually take to be \(b = \Theta(\sqrt{\log n})\) ), and is $(n/2) \log n$ up to lower-order terms. 

Finally, the with-replacement upper bound in
\cref{sec:upper-bound-with-replacement} uses a different block strategy.  We take
blocks of size $r \approx \sqrt{\log n}$
and run the classical patience-sorting algorithm (see, e.g.,~\cite{AldousDiaconis99}) inside each block,~i.e., a new value replaces the
leftmost stored value larger than it, or is appended to the right if no such stored value
exists.  The speedup over the without-replacement strategy comes from the fact
that a block no longer has to wait for samples in small edge windows.  Every
sample falling in the block's value interval can improve the local configuration.
After \(K\) local samples, a block is full as soon as those samples contain an
increasing subsequence of length \(r\).  Thus the failure probability is
controlled by
\[
\P(\LIS(\pi_K)<r)
\le
\frac{(r-1)^{2K}}{K!}
\le
\left(\frac{er^2}{K}\right)^K,
\]
where $\LIS(\pi_K)$ denotes the length of a longest increasing subsequence of a uniform permutation of length $K$. 
With \(r^2\asymp\log n\), taking \(K \asymp \log n\) makes this probability
polynomially small in \(n\), which suffices for a union bound. Since a
block sees local samples with probability about
\[
\frac{r}{n}\asymp \frac{\sqrt{\log n}}{n},
\]
collecting \(K\asymp\log n\) local samples takes about
\[
\frac{K}{r/n}\asymp n\sqrt{\log n}
\]
global samples.  
\subsection{Organization}

In \cref{sec:lower} we prove the lower bound in \cref{thm:main}, in \cref{sec:upper} we prove the upper bound in \cref{thm:main}, and in \cref{sec:upper-bound-with-replacement}, we prove the with-replacement upper bound \cref{thm:replacement-upper}.

\subsection{Notation}
All logarithms are natural. For $x \in \mb{R}$, we denote $(x)_+ := \max\{x,0\}$.
We will also make use of asymptotic notation. All asymptotic notation is with respect to $n \to \infty$. For functions $f,g$, $f = O(g)$ means that $\limsup_{n\to\infty} f/g < \infty$; $f = \Omega(g)$  means that $\liminf_{n\to\infty} f/g > 0$; and $f = \Theta(g)$ means that both $f = O(g)$ and $f = \Omega(g)$ hold.% For parameters $\varepsilon, \delta$, we write $\varepsilon \ll \delta$ to mean that $\varepsilon \le c(\delta)$ for a sufficient function $c$. A chain $\alpha \ll \beta \ll \gamma$ should be read from right to left.  

\subsection{Acknowledgments} The algorithm in \cref{thm:replacement-upper} was discovered by ChatGPT Pro. The same model was also used for various aspects of research, including literature overview, and assisted in manuscript preparation. V.J.~is partially supported by NSF grant DMS-2237646. C.M.~is partially supported by a Simons Dissertation Fellowship. Most of the research was conducted when D.K. was visiting the University of Illinois Chicago.

\section{Lower bound}
\label{sec:lower}

In this section we prove the lower bound in \cref{thm:main}. The argument tracks a potential associated with the current partial array state. Each maximal block \(B\) of contiguous unfilled array positions has a capacity \(c(B)\), namely the number of empty positions in the block, and an admissible value interval \(I(B)\), namely the set of values that can still be placed in that block without violating monotonicity. The potential penalizes blocks for which the capacity is large compared to the available length \(|I(B)|\). The main observation is that this potential cannot decrease under interior placements: it can decrease only when a sample is placed at the edge of an active block, or when a singleton block is completed. The set of values for which this can happen has small total measure.  This gives a quantitative bound on the rate at which the potential can decay, and hence a lower bound on the completion time.

\medskip

We begin with some notation. Fix a possibly randomized adaptive strategy $S$ and let~$K_t$ denote the number of accepted values after~$t$ turns. Let $\tau$ be the (random) finishing time, i.e.
\[
\tau:=\inf\{t\ge 0:K_t=n\}.
\]
At time $t$, the empty~$*$ array positions form a disjoint union of maximal consecutive intervals, each of which we refer to as a block, and we let~$\cB_t$ denote the collection of blocks after~$t$ turns. For each such block $B$:
\begin{itemize}
\item $c(B)$ denotes its capacity, i.e.~the number of empty positions in that block,
\item $I(B)\subseteq [0,1]$ denotes the interval of values which may be legally placed in $B$; this is determined by the nearest filled entries to the left and right of $B$ (or by the boundaries $0$ and $1$), and
\item $L(B):=|I(B)|$ denotes the length of this interval.
\end{itemize}
The intervals~$I(B), B \in \mc{B}_t$ are disjoint, up to endpoints. An example is given in \cref{fig:partial-state}.

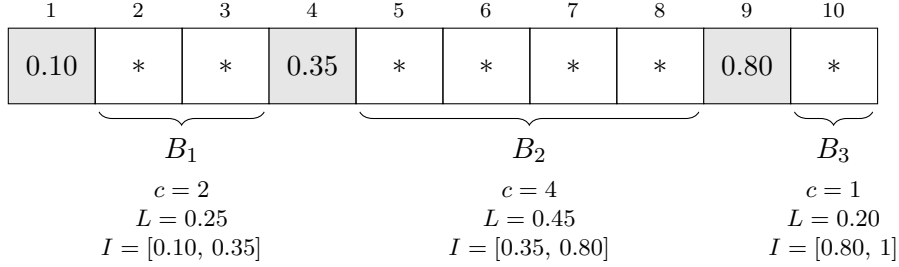
\begin{figure}[ht]
\centering
\begin{tikzpicture}[
  x=1.15cm,
  cell/.style={draw, minimum width=1.15cm, minimum height=1cm, anchor=south west, inner sep=0pt},
  filled/.style={cell, fill=gray!20},
  empty/.style={cell}
]
  % Array cells (positions 1..10), filled at 1, 4, 9; left endpoint filled
  \node[filled] at (0,0) {$0.10$};
  \node[empty]  at (1,0) {$\ast$};
  \node[empty]  at (2,0) {$\ast$};
  \node[filled] at (3,0) {$0.35$};
  \node[empty]  at (4,0) {$\ast$};
  \node[empty]  at (5,0) {$\ast$};
  \node[empty]  at (6,0) {$\ast$};
  \node[empty]  at (7,0) {$\ast$};
  \node[filled] at (8,0) {$0.80$};
  \node[empty]  at (9,0) {$\ast$};

  % Position indices
  \foreach \i in {1,...,10} {
    \node[font=\scriptsize, above] at (\i-0.5, 1.02) {$\i$};
  }

  % Block braces: B_1 = positions 2-3, B_2 = positions 5-8, B_3 = position 10
  \draw[decorate, decoration={brace, amplitude=5pt, mirror}]
    (1.05, -0.1) -- node[below=6pt] {$B_1$} (2.95, -0.1);
  \draw[decorate, decoration={brace, amplitude=5pt, mirror}]
    (4.05, -0.1) -- node[below=6pt] {$B_2$} (7.95, -0.1);
  \draw[decorate, decoration={brace, amplitude=5pt, mirror}]
    (9.05, -0.1) -- node[below=6pt] {$B_3$} (9.95, -0.1);

  % Per-block parameter labels
  \node[align=center, font=\footnotesize] at (2.0, -1.55)
    {$c=2$\\$L=0.25$\\$I=[0.10,\,0.35]$};
  \node[align=center, font=\footnotesize] at (6.0, -1.55)
    {$c=4$\\$L=0.45$\\$I=[0.35,\,0.80]$};
  \node[align=center, font=\footnotesize] at (9.5, -1.55)
    {$c=1$\\$L=0.20$\\$I=[0.80,\,1]$};
\end{tikzpicture}
\caption{A partial array state for $n=10$ in which the three filled values $0.10<0.35<0.80$ occupy positions $1,4,9$. There are three blocks $B_1,B_2,B_3$, with~$c=c(B)$,~$I=I(B)$, and ~$L=L(B)=|I(B)|$ displayed.}
\label{fig:partial-state}
\end{figure}

\medskip

When a new sample is accepted for placement in the array, there are three possibilities: 
\begin{itemize}
\item an \emph{interior move} places the sample in a position which splits the chosen block into two nonempty child blocks;
\item an \emph{edge move} places the sample in the leftmost or rightmost empty position of a block of capacity at least \(2\); 
\item a \emph{singleton completion} fills a block of capacity \(1\).
\end{itemize}
In the example in \cref{fig:partial-state}, a placement in positions~$6,7$ would be an interior move, in positions~$2,3,5,8$ an edge move, and in position~$10$  would be a singleton completion.

\medskip

Fix once and for all a parameter $a\in(0,n+1)$. For a block $B$, define
\[
q_a(B):= \bigl(c(B)+1-aL(B)\bigr)_+,
\]
and set
\[
Q_t:=\sum_{B \in \cB_t} q_a(B).
\]
Thus \(q_a(B)\) penalizes a block whose capacity is large relative to the length of its admissible interval.

\begin{lemma}\label{lem:easy_properties}
The potential $Q_t$ has the following properties.
\begin{enumerate}
\item for every $t\ge 0$,
\[
0\le Q_t\le n+1,
\]
\item $Q_0= n+1-a $ and $Q_\tau=0$,
\item an interior move does not decrease $Q_t$,
\item an edge move decreases $Q_t$ by at most $1$, and
\item a singleton completion decreases $Q_t$ by at most $2$.
\end{enumerate}
\end{lemma}

\begin{proof}
We prove the assertions in order. 
\begin{enumerate}
    \item First, notice that for any~$t \geq 0$ we have that
    \[
    |\cB_t| + \sum_{B \in \cB_t} c(B) \leq n+1.
    \]
    Also \(0\le q_a(B)\le c(B)+1\) for every \(B\in\cB_t\).
Hence
\[
0\le Q_t\le \sum_{B\in\cB_t}(c(B)+1)\le n+1.
\]
\item Initially~$\cB_0$ consists of a single block $B$ with $c(B)=n$ and $L(B)=1$, so
\[
Q_0= (n+1-a)_+ = n+1-a.
\]
At time~$\tau$,~$\cB_\tau = \emptyset$ and hence $Q_\tau=0$.
\item Consider an interior move occurring in a block~$B \in \cB_t$; let the child blocks be $B_1,B_2$. Then
\[
c(B_1)+c(B_2)=c(B) - 1,\qquad L(B_1)+L(B_2)=L(B).
\]
Writing
\[
y_i:=c(B_i)+1-aL(B_i)\quad (i=1,2),\qquad y:=c(B)+1-aL(B),
\]
we have $y_1+y_2=y$. Therefore
\[
q_a(B_1)+q_a(B_2)
= (y_1)_+ + (y_2)_+
\ge (y_1 + y_2)_+
= y_+
= q_a(B),
\]
so an interior move cannot decrease $Q_t$.

\item Consider an edge move occurring in a block~$B \in \cB_t$, and denote the single child block by~$B'$. Let~$x = L(B) - L(B')\geq 0$. Then $B'$ satisfies
\[
c(B')=c(B)-1,\qquad L(B')=L(B)-x,
\]
so
\begin{equation}\label{eq:change-in-q}
c(B')+1-aL(B')=\bigl(c(B)+1-aL(B)\bigr)-1+ax.
\end{equation}
Since \(ax\ge0\), taking positive parts gives
\[
q_a(B')\ge q_a(B)-1.
\]
Thus an edge move decreases \(Q_t\) by at most \(1\).

\item Consider a singleton completion move which deletes a block~$B \in \cB_t$. In this case~$c(B)=1$, so
\[
q_a(B)=\bigl(2-aL(B)\bigr)_+\le 2.
\]
Completion removes $B$ entirely, hence it decreases $Q_t$ by at most $2$. \qedhere
\end{enumerate}
\end{proof}

The next lemma is the key estimate. The sequence~$Q_t$ must decrease towards~$Q_\tau = 0$, but can only do so in steps of size at most~$2$. In fact,~$Q_t$ can only decrease at all when the player makes edge moves or completes singleton blocks. The point is that, conditioned on the present state, the set of sample values which can trigger such a move has total measure \(O(Q_t/a)\). More precisely, let  $(\mathcal{F}_t)_{t\ge 0}$ be the natural filtration generated by the samples revealed through time \(t\), together with any auxiliary randomness of the strategy used through time \(t\).

% The following lemma controls the decay of $Q_t$.

\begin{lemma}\label{lem:expected_positive_decrement}
For every \(t\ge0\) and every possibly randomized adaptive strategy,
\[
\E\bigl[(Q_t-Q_{t+1})_+\mid \mathcal F_t\bigr]
\le
\frac{2Q_t}{a}.
\]
\end{lemma}

\begin{proof}
Condition on \(\mathcal F_t\), so that the current state is fixed. Throughout this proof, all expectations are $\mb{E}[\cdot\mid\mc{F}_t]$; we omit the dependence on $\mc{F}_t$. Set
\[
D_t:=(Q_t-Q_{t+1})_+.
\]
It suffices to prove that, for every block $B \in \mc{B}_t$,
\begin{equation}\label{eq:one-block-control-on-decrement} \mb{E}[D_t \mbf1\{ X_{t+1} \in I(B) \} ] \leq \frac{2q_a(B)}{a}. \end{equation}
Indeed, since the intervals \(I(B)\), \(B\in\mc B_t\), are disjoint up to endpoints, and since \(X_{t+1}\) has a continuous distribution, the events \(\{X_{t+1}\in I(B)\}\) partition the event that \(X_{t+1}\) lies in an admissible interval, up to a null set. Therefore
\[
\E[D_t\mid \mc F_t]
=
\sum_{B\in\mc B_t}
\E\bigl[D_t\mbf1\{X_{t+1}\in I(B)\}\mid \mc F_t\bigr].
\]
Combining this identity with \eqref{eq:one-block-control-on-decrement} gives
\[
\E[D_t\mid \mc F_t]
\le
\sum_{B\in\mc B_t}\frac{2q_a(B)}{a}
=
\frac{2Q_t}{a},
\]
as desired.

We now prove \eqref{eq:one-block-control-on-decrement}. Let $B \in \mc{B}_t$. First suppose that \(c(B)=1\).  On the event \(X_{t+1}\in I(B)\), a placement in \(B\) can remove at most the contribution \(q_a(B)\).  Thus
\[
D_t\mbf1\{X_{t+1}\in I(B)\}
\le
q_a(B)\mbf1\{X_{t+1}\in I(B)\}
\]
pointwise.  Hence
\[
\E\bigl[D_t\mbf1\{X_{t+1}\in I(B)\}\bigr]
\le
q_a(B)L(B).
\]
If \(q_a(B)=0\), this already proves \eqref{eq:one-block-control-on-decrement}.  If \(q_a(B)>0\), then \(2-aL(B)>0\), and therefore \(L(B)<2/a\).  This again gives \eqref{eq:one-block-control-on-decrement}.

Now suppose that \(c(B)\ge2\), and let \(I(B)=[\ell,u]\).  By \cref{lem:easy_properties}, the potential can decrease only if the strategy makes an edge move in \(B\).  Suppose such a move is made on the left edge, and write \(X_{t+1}=\ell+x\).  By \eqref{eq:change-in-q}, if \(x\ge1/a\), then the block contribution cannot decrease.  Thus a left-edge move which decreases \(Q_t\) must have \(x<1/a\).  The same argument applies at the right edge.  Consequently, if a placement in \(B\) decreases \(Q_t\), then
\[
X_{t+1}\in
\left[\ell,\ell+\frac1a\right)
\cup
\left(u-\frac1a,u\right].
\]
Moreover, the decrease caused by \(B\) is at most \(q_a(B)\).  Therefore
\[
\E\bigl[D_t\mbf1\{X_{t+1}\in I(B)\}\bigr]
\le
q_a(B)\,
\P\left(
X_{t+1}\in
\left[\ell,\ell+\frac1a\right)
\cup
\left(u-\frac1a,u\right]
\right)
\le
\frac{2q_a(B)}{a}. \qedhere
\]

\end{proof}

We now convert the one-step decrement estimate into a lower bound on the time
needed to drive the potential to zero.  Define
\[
\Phi(q):=\int_0^q \frac{ds}{s+2}
=
\log\left(1+\frac q2\right),
\qquad q\ge0.
\]

From now on we extend the process after completion by setting \(Q_t=0\) for all \(t\ge \tau\).

\begin{lemma}\label{lem:phi-progress}
For every integer
\(N\ge1\),
\begin{equation}\label{eq:finite-progress-expectation}
\Phi(Q_0)-\E[\Phi(Q_N)]
\le
\frac2a\sum_{t=0}^{N-1}\P(t<\tau).
\end{equation}
Moreover, for every integer \(N\ge1\) and every \(\theta>0\),
\begin{equation}\label{eq:finite-progress-tail}
\P(\tau\le N)
\le
\exp\left(
-\theta\Phi(Q_0)
+
\frac{2N}{a}(e^\theta-1)
\right).
\end{equation}
\end{lemma}

\begin{proof}
Set
\[
D_t:=(Q_t-Q_{t+1})_+,
\qquad
Z_t:=\bigl(\Phi(Q_t)-\Phi(Q_{t+1})\bigr)_+.
\]
We first prove the one-step estimate
\begin{equation}\label{eq:Z-one-step}
\E[Z_t\mid \mc F_t]\le \frac2a\,\mbf1_{\{t<\tau\}}.
\end{equation}
If \(Q_t=0\), then \(Q_{t+1}\ge0\), and hence
\(\Phi(Q_t)-\Phi(Q_{t+1})\le0\). Thus \(Z_t=0\). If \(Q_t>0\), then
\[
Z_t
\le
\Phi(Q_t)-\Phi(Q_t-D_t).
\]
By \cref{lem:easy_properties}, a single step can decrease \(Q_t\) by at most \(2\).  Hence every
\(s\in[Q_t-D_t,Q_t]\) satisfies \(s+2\ge Q_t\), and therefore
\[
\Phi(Q_t)-\Phi(Q_t-D_t)
=
\int_{Q_t-D_t}^{Q_t}\frac{ds}{s+2}
\le
\frac{D_t}{Q_t}.
\]
Taking conditional expectations and using \cref{lem:expected_positive_decrement} gives
\[
\E[Z_t\mid \mc F_t]
\le
\frac1{Q_t}\E[D_t\mid \mc F_t]
\le
\frac2a
\]
on the event \(\{Q_t>0\}\), which is contained in \(\{t<\tau\}\).  This proves \eqref{eq:Z-one-step}.

Now
\[
\Phi(Q_t)-\Phi(Q_{t+1})\le Z_t,
\]
so summing \eqref{eq:Z-one-step} from \(t=0\) to \(N-1\) gives
\[
\Phi(Q_0)-\E[\Phi(Q_N)]
=
\sum_{t=0}^{N-1}\E[\Phi(Q_t)-\Phi(Q_{t+1})]
\le
\sum_{t=0}^{N-1}\E[Z_t]
\le
\frac2a\sum_{t=0}^{N-1}\P(t<\tau).
\]
This proves \eqref{eq:finite-progress-expectation}.

It remains to prove the tail estimate.  Since \(\Phi\) is \(1/2\)-Lipschitz and \(Q_t\)
can decrease by at most \(2\) in one step, we have \(0\le Z_t\le1\).  Thus, for every
\(\theta>0\),
\[
e^{\theta Z_t}\le 1+(e^\theta-1)Z_t.
\]
Using \eqref{eq:Z-one-step}, and then dropping the indicator, gives
\[
\E[e^{\theta Z_t}\mid \mc F_t]
\le
1+(e^\theta-1)\E[Z_t\mid \mc F_t]
\le
\exp\left(\frac2a(e^\theta-1)\right).
\]
Iterating this, we have
\[
\E\exp\left(\theta\sum_{t=0}^{N-1}Z_t\right)
\le
\exp\left(\frac{2N}{a}(e^\theta-1)\right).
\]
On the event \(\{\tau\le N\}\), we have \(Q_N=0\), so that
\[
\Phi(Q_0)
=
\sum_{t=0}^{N-1}\bigl(\Phi(Q_t)-\Phi(Q_{t+1})\bigr)
\le
\sum_{t=0}^{N-1}Z_t.
\]
Therefore, by Markov's inequality, \[
\P(\tau\le N)
\le
\exp(-\theta\Phi(Q_0))
\E\exp\left(\theta\sum_{t=0}^{N-1}Z_t\right),
\]
which proves \eqref{eq:finite-progress-tail}.
\end{proof}

We now have all the ingredients to prove the lower bound. We first prove the lower bound in expectation and then prove the high-probability version.

\begin{proposition}\label{prop:lower}
For every possibly randomized adaptive strategy \(S\) with completion time \(\tau\),
\begin{equation}
\label{eq:lower-bound-sup-version}
\E[\tau]
\ge
\sup_{0<a<n+1}
\frac a2
\log\left(1+\frac{n+1-a}{2}\right) \geq \left(\frac12- o(1)\right)n\log n.
\end{equation}
\end{proposition}

\begin{proof}
If \(\E[\tau]=\infty\), there is nothing to prove.  Hence assume that
\(\E[\tau]<\infty\).  In particular, \(\tau<\infty\) almost surely. Applying \eqref{eq:finite-progress-expectation} and then letting \(N\to\infty\), we get
\[
\Phi(Q_0)
\le
\frac2a\,\E[\tau].
\]
Indeed, \(0\le Q_N\le n+1\) for all \(N\), while \(Q_N\to0\) almost surely because
the process is finished after the almost surely finite time \(\tau\).  Therefore
\(\E[\Phi(Q_N)]\to0\) by dominated convergence.  Also,
\[
\sum_{t=0}^\infty \P(t<\tau)=\E[\tau].
\]
By \cref{lem:easy_properties}, \(Q_0=n+1-a\).  Hence
\[
\E[\tau]
\ge
\frac a2 \Phi(Q_0)
=
\frac a2
\log\left(1+\frac{n+1-a}{2}\right).
\]
Taking the supremum over $a$ gives the first inequality in \eqref{eq:lower-bound-sup-version}. The second inequality follows from the choice $a = n+1 - n/\log n$.
\end{proof}

\begin{proposition}\label{prop:hp-lower}
For every fixed \(\eps\in(0,1/2)\), there exists \(n_0(\eps)\) such that, for all
\(n\ge n_0(\eps)\) and every possibly randomized adaptive strategy \(S\) with
completion time \(\tau\),
\[
\P\left(\tau\le \left(\frac12-\eps\right)n\log n\right)
\le
n^{-2\eps^2}.
\]
\end{proposition}

\begin{proof}
Let
\[
N=\left\lfloor\left(\frac12-\eps\right)n\log n\right\rfloor,
\qquad
a=n+1-\frac{n}{\log n}.
\]
By \eqref{eq:finite-progress-tail} and \cref{lem:easy_properties}, \(Q_0=n+1-a\), so
for every \(\theta>0\),
\[
\P(\tau\le N)
\le
\exp\left(
-\theta\log\left(1+\frac{n+1-a}{2}\right)
+
\frac{2N}{a}(e^\theta-1)
\right).
\]
Set
\[
\theta=-\log(1-2\eps).
\]
Then
\[
\log\left(1+\frac{n+1-a}{2}\right)
=
\log\left(1+\frac{n}{2\log n}\right)
=
\log n - \log(2\log n) + o(1),
\]
while
\[
\frac{2N}{a}
=
(1-2\eps)\log n+O_\eps(1).
\]
Since \(e^\theta-1=2\eps/(1-2\eps)\), it follows that
\[
\P(\tau\le N)
\le
\exp\left(
\bigl(\log(1-2\eps)+2\eps\bigr)\log n
+
O_\eps(\log\log n)
\right).
\]
Finally,
\[
-\log(1-2\eps)-2\eps>2\eps^2
\qquad (0<\eps<1/2),
\]
so the last display is at most \(n^{-2\eps^2}\) for all sufficiently large \(n\).
\end{proof}

\section{Upper bound}\label{sec:upper}

In this section we prove the upper bound in \cref{thm:main}. The strategy is a block version of the coupon-collector strategy. 
Accordingly, fix an integer \(b\) with \(1\le b\le \sqrt n\), to be chosen at the end of the proof. Let
\[
m:=\floor{\frac nb}.
\]
Write \(n=mb+r\), where \(0\le r<b\). Since \(b\le \sqrt n\), we have \(r\le b-1\le m-1\). Thus we may choose block sizes
\[
s_1,\dots,s_m\in\{b,b+1\}
\]
so that \(\sum_{j=1}^m s_j=n\). 
Partition the coordinates \(\{1,\dots,n\}\) into consecutive blocks
\[
B_1,\dots,B_m,
\qquad
|B_j|=s_j,
\]
and partition \([0,1]\) into consecutive intervals
\[
J_1,\dots,J_m,
\qquad
|J_j|=\frac{s_j}{n}.
\]

The strategy $S_b$ is formally described in \cref{alg:block-strategy}; we present an informal description here. The strategy simultaneously plays the game in each block $B_j$ with entries in $J_j$.  Since the blocks $B_j$ and the intervals $J_j$ are ordered in the same way, local monotonicity within each block will ensure global monotonicity. For each block \(B_j\), the strategy maintains two pieces of local data. The first is an integer \(k_j\), the number of currently unfilled locations in \(B_j\). The second is a feasible interval
\[
I_j^{\rm feas}=[L_j,R_j]\subseteq J_j.
\]
Initially,
\[
k_j=s_j,
\qquad
I_j^{\rm feas}=J_j.
\]
The invariant is that the unfilled locations in \(B_j\) form a contiguous run, and future values placed into this run must lie in \(I_j^{\rm feas}\). Suppose a new sample $x\in J_j$ arrives while the current state of block $B_j$ is $(k,I)$ with $k\geq 1$ and $I=[L,R]$. Write $\ell:=R-L$. The strategy accepts $x$ if and only if it lies in one of the two edge intervals
\[
\Bigl[L,L+\frac{\ell}{k+1}\Bigr]
\qquad\text{or}\qquad
\Bigl[R-\frac{\ell}{k+1},R\Bigr].
\]
If $x$ lies in the left edge interval, it is placed in the leftmost currently empty location of $B_j$, and the new feasible interval becomes $[x,R]$. If $x$ lies in the right edge interval, it is placed in the rightmost currently empty location of $B_j$, and the new feasible interval becomes $[L,x]$. If $x$ lies outside these edge intervals, it is discarded. See \cref{fig:block-strategy} for an illustration.

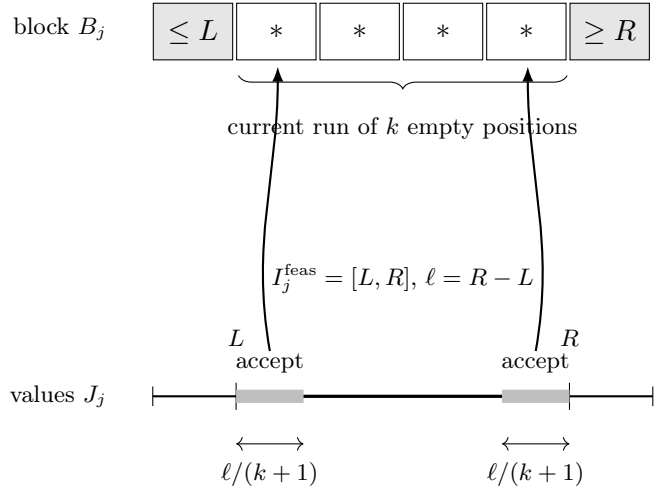
\begin{figure}[ht]
\centering
\begin{tikzpicture}[
  x=1.05cm,
  y=1cm,
  cell/.style={draw, minimum width=1.05cm, minimum height=.75cm, anchor=south west, inner sep=0pt},
  filled/.style={cell, fill=gray!20},
  empty/.style={cell},
  lab/.style={font=\footnotesize},
  arr/.style={-{Latex[length=2mm]}, thick}
]

\node[lab, anchor=east] at (0.55,5.92) {block \(B_j\)};

\node[filled] at (1,5.50) {\(\le L\)};
\node[empty]  at (2.05,5.50) {\(\ast\)};
\node[empty]  at (3.10,5.50) {\(\ast\)};
\node[empty]  at (4.15,5.50) {\(\ast\)};
\node[empty]  at (5.20,5.50) {\(\ast\)};
\node[filled] at (6.25,5.50) {\(\ge R\)};

\draw[decorate, decoration={brace, amplitude=5pt, mirror}]
  (2.08,5.30) -- node[below=12pt, lab] {current run of \(k\) empty positions} (6.22,5.30);

\node[lab, anchor=east] at (0.55,1.05) {values $J_j$};

\draw[thick] (1,1.05) -- (7.30,1.05);
\draw (1,1.18) -- (1,0.92);
\draw (7.30,1.18) -- (7.30,0.92);

\draw[line width=1.2pt] (2.05,1.05) -- (6.25,1.05);
\draw (2.05,1.23) -- (2.05,0.87);
\draw (6.25,1.23) -- (6.25,0.87);

\node[above=10pt, lab] at (2.05,1.23) {\(L\)};
\node[above=10pt, lab] at (6.25,1.23) {\(R\)};
\node[above=30pt, lab] at (4.15,1.23) {\(I_j^{\rm feas}=[L,R]\), \(\ell=R-L\)};

\draw[line width=5pt, gray!50] (2.05,1.05) -- (2.90,1.05);
\draw[line width=5pt, gray!50] (5.40,1.05) -- (6.25,1.05);

\draw[<->] (2.05,0.42) -- node[below=3pt, lab] {\(\ell/(k+1)\)} (2.90,0.42);
\draw[<->] (5.40,0.42) -- node[below=3pt, lab] {\(\ell/(k+1)\)} (6.25,0.42);

\node[lab, above=5pt] at (2.475,1.05) {accept};
\node[lab, above=5pt] at (5.825,1.05) {accept};

\draw[arr] (2.475,1.65) to[out=100,in=-90] (2.575,5.42);
\draw[arr] (5.825,1.65) to[out=80,in=-90] (5.725,5.42);

\end{tikzpicture}

\caption{The local rule inside a block \(B_j\).  When \(k\) locations remain unfilled, the block has a feasible value interval \(I_j^{\rm feas}=[L,R]\) of length \(\ell=R-L\).  The strategy accepts a sample only if it lies in one of the two gray edge intervals, each of length \(\ell/(k+1)\).  A sample from the left edge interval fills the leftmost empty location and changes the feasible interval to \([x,R]\); a sample from the right edge interval fills the rightmost empty location and changes it to \([L,x]\).}
\label{fig:block-strategy}
\end{figure}

\begin{algorithm}[t]
\caption{The strategy \(S_b\)}
\label{alg:block-strategy}
\DontPrintSemicolon
Partition the coordinates into \(B_1,\dots,B_m\) and \([0,1]\) into \(J_1,\dots,J_m\) as above.\;
\ForEach{\(j\in\{1,\dots,m\}\)}{
Set \(k_j\gets s_j\).\;
Set \(I_j^{\rm feas}\gets J_j\).\;
}
\ForEach{new sample \(x\)}{
Let \(j\) be such that \(x\in J_j\), with endpoint ties broken arbitrarily.\;
\If{\(k_j=0\)}{
Discard \(x\).\;
}
\Else{
Write \(I_j^{\rm feas}=[L,R]\) and \(\ell=R-L\).\;
\If{\(x\in [L,L+\ell/(k_j+1)]\)}{
Place \(x\) in the leftmost currently empty location of \(B_j\).\;
Set \(I_j^{\rm feas}\gets [x,R]\) and \(k_j\gets k_j-1\).\;
}
\ElseIf{\(x\in [R-\ell/(k_j+1),R]\)}{
Place \(x\) in the rightmost currently empty location of \(B_j\).\;
Set \(I_j^{\rm feas}\gets [L,x]\) and \(k_j\gets k_j-1\).\;
}
\Else{
Discard \(x\).\;
}
}
}
\end{algorithm}

\begin{lemma}\label{lem:strategy-admissible}
The strategy \(S_b\) is well defined and preserves global monotonicity of the array.
\end{lemma}

\begin{proof}
It suffices to check the following invariant inside each block \(B_j\). At every time, the unfilled locations of \(B_j\) form a contiguous run, and if \(I_j^{\rm feas}=[L,R]\), then all filled entries to the left of this run are at most \(L\), while all filled entries to the right are at least \(R\).

The invariant is trivial initially. Suppose it holds when \(k_j\ge1\) locations remain. If \(x\) is accepted from the left edge interval, then \(x\in[L,R]\), the strategy places \(x\) in the leftmost empty location, and updates the feasible interval to \([x,R]\). The invariant is preserved, since the old left entries are at most \(L\le x\), while the right entries are still at least \(R\). The right-edge case is identical, with the feasible interval updated to \([L,x]\). Discarded samples do not change the state.

Thus each block remains internally monotone. Since the intervals \(J_1,\dots,J_m\) and the coordinate blocks \(B_1,\dots,B_m\) are ordered from left to right, internal monotonicity in each block implies global monotonicity of the whole array.
\end{proof}

The next lemma is the key quantitative invariant used in the analysis. It says that each unfinished block always has an acceptance region of length at least \((2+o(1))/n\).

\begin{lemma}\label{lem:feasible-length}
Fix a block \(B_j\) of size \(s=s_j\). Suppose that \(k\ge1\) locations remain unfilled in this block, and write the current feasible interval as
\[
I_j^{\rm feas}=[L,R],
\qquad
\ell:=R-L.
\]
Then
\[
\ell\ge \frac{k+1}{s+1}\cdot \frac{s}{n}.
\]
Consequently, if
\[
E^-:=\left[L,L+\frac{\ell}{k+1}\right],
\qquad
E^+:=\left[R-\frac{\ell}{k+1},R\right]
\]
are the two edge intervals from which the strategy accepts samples, then
\[
|E^-|+|E^+|
=
\frac{2\ell}{k+1}
\ge
\frac{2s}{(s+1)n}
\ge
\frac{2b}{(b+1)n}.
\]
\end{lemma}

\begin{proof}
Let \(\ell_k\) denote the feasible-interval length when \(k\) locations remain in the block. We prove by downward induction on \(k\) that
\[
\ell_k\ge \frac{k+1}{s+1}\cdot \frac{s}{n}.
\]
For \(k=s\), the feasible interval is \(J_j\), whose length is \(s/n\), so the claim holds.

Now suppose the claim holds for some \(2\le k\le s\). If the next accepted sample in this block arrives while \(k\) locations remain, then it lies in one of the two edge intervals, each of length \(\ell_k/(k+1)\). After this placement, the new feasible interval therefore has length at least
\[
\ell_k-\frac{\ell_k}{k+1}
=
\frac{k}{k+1}\ell_k.
\]
By the induction hypothesis,
\[
\ell_{k-1}
\ge
\frac{k}{k+1}\ell_k
\ge
\frac{k}{s+1}\cdot \frac{s}{n},
\]
which is the desired bound with \(k-1\) in place of \(k\).

The total length of the two edge intervals when \(k\) locations remain is \(2\ell_k/(k+1)\). Using the first part gives
\[
\frac{2\ell_k}{k+1}
\ge
\frac{2s}{(s+1)n}.
\]
Since \(s\in\{b,b+1\}\), we have \(s/(s+1)\ge b/(b+1)\), proving the final inequality.
\end{proof}

We now estimate the completion time. The key point is that conditioned on the past, every unfinished block has probability at least \(2b/((b+1)n)\) of accepting the next sample.

\begin{proposition}\label{prop:upper-explicit}
For each \(n\), there is a deterministic strategy \(S=S^{(n)}\) such that
\[
\mb{E}[\tau_{S}]\le \left(\frac12+o(1)\right)n\log n.
\]
Moreover, for every fixed \(\eps>0\),
\[
\P\left(\tau_{S}>\left(\frac12+\eps\right)n\log n\right)\le n^{-\eps}
\]
for all sufficiently large \(n\).
\end{proposition}

\begin{proof}
Set
\[
b:=\left\lceil \sqrt{\log n}\right\rceil
\]
and define \(S^{(n)}:=S_b\). Let \(\tau_b\) denote the completion time of this strategy, and set
\[
\lambda:=\frac{2b}{(b+1)n}.
\]
By \cref{lem:feasible-length}, whenever a block is unfinished, the current acceptance region for that block has length at least \(\lambda\).

Fix a block \(B_j\), and let \(\tau_j\) be the time at which \(B_j\) is completed. If \(s_j=|B_j|\), write
\[
\tau_j=W_1^{(j)}+\cdots+W_{s_j}^{(j)},
\]
where \(W_r^{(j)}\) is the number of samples between the \((r-1)\)-st and \(r\)-th accepted samples in \(B_j\). Condition on the history up to the time at which \(r-1\) samples have been accepted in \(B_j\). Until the next acceptance in \(B_j\), every new sample has conditional probability at least \(\lambda\) of being accepted by this block. Therefore \(W_r^{(j)}\) is conditionally dominated by a geometric random variable \(G_\lambda\) with
\[
\P(G_\lambda=q)=(1-\lambda)^{q-1}\lambda,
\qquad q\ge1.
\]
Consequently, for every \(0<\theta<-\log(1-\lambda)\),
\[
\E\left[e^{\theta W_r^{(j)}}\mid \mc F_r \right]
\le
\E e^{\theta G_\lambda}
=
\frac{\lambda e^\theta}{1-(1-\lambda)e^\theta},
\]
where~$\mc F_r$ is the natural filtration on all past events.
Iterating over \(r=1,\dots,s_j\), and using \(s_j\le b+1\), gives
\begin{equation}\label{eq:block-mgf-discrete}
\E e^{\theta\tau_j}
\le
\left(\frac{\lambda e^\theta}{1-(1-\lambda)e^\theta}\right)^{b+1}.
\end{equation}

Now choose
\[
\eta:=\frac1{\log n},
\qquad
\theta:=-\log\bigl(1-(1-\eta)\lambda\bigr).
\]
Then
\[
e^\theta=\frac{1}{1-(1-\eta)\lambda},
\]
and hence
\[
\frac{\lambda e^\theta}{1-(1-\lambda)e^\theta}=\frac1\eta.
\]
Thus \eqref{eq:block-mgf-discrete} gives
\begin{equation}\label{eq:block-mgf-specialized}
\E e^{\theta\tau_j}\le \eta^{-(b+1)}
\end{equation}
for every block \(B_j\).

Since the array is complete exactly when all blocks are complete,
\[
\tau_b=\max_{1\le j\le m}\tau_j.
\]
Using
\[
\max_{1\le j\le m} x_j
\le
\frac1\theta\log\left(\sum_{j=1}^m e^{\theta x_j}\right),
\qquad x_1,\dots,x_m\in\R,
\]
together with Jensen's inequality and \eqref{eq:block-mgf-specialized}, we obtain
\[
\E[\tau_b]
\le
\frac1\theta
\log\left(\sum_{j=1}^m \E e^{\theta\tau_j}\right)
\le
\frac1\theta
\left(
\log m+(b+1)\log\frac1\eta
\right).
\]
Since
\[
\theta
=
-\log\bigl(1-(1-\eta)\lambda\bigr)
\ge
(1-\eta)\lambda
=
(1-\eta)\frac{2b}{(b+1)n},
\]
it follows that
\[
\E[\tau_b]
\le
\frac{(b+1)n}{2b(1-\eta)}
\left(
\log m+(b+1)\log\frac1\eta
\right).
\]
For our choices \(b=\lceil\sqrt{\log n}\rceil\) and \(\eta=1/\log n\), we have
\[
\frac{b+1}{b}=1+o(1),
\qquad
\frac1{1-\eta}=1+o(1),
\qquad
\log m=\log n+o(\log n),
\]
and
\[
(b+1)\log\frac1\eta
=
O(\sqrt{\log n}\log\log n)
=
o(\log n).
\]
Therefore
\[
\E[\tau_b]
\le
\left(\frac12+o(1)\right)n\log n.
\]

It remains to prove the high-probability bound. By Markov's inequality and \eqref{eq:block-mgf-specialized}, for every \(t\ge0\),
\begin{equation}\label{eq:block-tail-specialized}
\P(\tau_b>t)
\le
e^{-\theta t}\sum_{j=1}^m \E e^{\theta\tau_j}
\le
\exp\left(
\log m+(b+1)\log\frac1\eta-\theta t
\right).
\end{equation}
Let
\[
t_n:=\left(\frac12+\eps\right)n\log n.
\]
Using \(\theta\ge (1-\eta)\lambda\), \eqref{eq:block-tail-specialized} gives
\begin{align*}
\log \P(\tau_b>t_n)
&\le
\log m+(b+1)\log\frac1\eta
-
(1-\eta)\frac{2b}{(b+1)n}
\left(\frac12+\eps\right)n\log n  \\
&\le
\log n+o(\log n)
-
(1+2\eps)(1-o(1))\log n  \\
&=
(-2\eps+o(1))\log n.
\end{align*}
For all sufficiently large \(n\), this is at most \(-\eps\log n\). Hence
\[
\P\left(\tau_b>\left(\frac12+\eps\right)n\log n\right)\le n^{-\eps},
\]
as claimed.
\end{proof}

\begin{remark}\label{rem:fixed-b}
For a given choice of \(b\), the above proof gives
\[
\mb{E}[\tau_{S_b}(n)]
\le
\frac{b+1}{2b}n\log n+O_b(n\log\log n).
\]
Thus \(b=1\) gives the coupon-collector constant \(1\), \(b=2\) gives the constant \(3/4\), and increasing the block size drives the leading constant down to \(1/2\).
\end{remark}

\section{With-replacement upper bound}
\label{sec:upper-bound-with-replacement}
In this section we prove \cref{thm:replacement-upper}. The array again begins empty and must remain monotone at all times, but the player may overwrite old entries.  After seeing \(X_t\), a strategy may either discard \(X_t\) or place \(X_t\) in some coordinate \(j\) so that
\[
(a_1,\ldots,a_{j-1},X_t,a_{j+1},\ldots,a_n)
\]
is a valid partial array state, regardless of whether \(a_j=\ast\).

The strategy is again block-based, but the local rule is different from the without-replacement rule in \cref{sec:upper}. The advantage of allowing replacements is that a block can continually improve its current configuration. Unlike the without-replacement strategy, progress does not depend on waiting for samples to fall into shrinking acceptance intervals.

Fix an integer parameter \(r\ge2\), to be chosen at the end of the proof.  We assume throughout that \(r\le \sqrt n\).  Let
\[
m:=\floor{\frac nr}.
\]
Write \(n=mr+q\), where \(0\le q<r\).  Since \(r\le \sqrt n\), we have \(q\le r-1\le m-1\).  Thus we may choose block sizes
\[
s_1,\dots,s_m\in\{r,r+1\}
\]
so that \(\sum_{j=1}^m s_j=n\).  Partition the coordinates \(\{1,\dots,n\}\) into consecutive blocks
\[
B_1,\dots,B_m,
\qquad
|B_j|=s_j,
\]
and partition \([0,1]\) into consecutive intervals
\[
I_1,\dots,I_m,
\qquad
|I_j|=\frac{s_j}{n}.
\]

Since the samples have continuous distributions, with probability one no two samples
falling in the same interval \(I_j\) are equal. Throughout the analysis we work on
this probability-one event. On the null event of a tie, the strategy may use any fixed
tie-breaking rule which preserves monotonicity; this choice has no effect on the
estimates below. The with-replacement strategy $S'_r$ is formally described in \cref{alg:with-replacement-block}; we present an informal description here. As before, the strategy only places values from \(I_j\) into the coordinate block \(B_j\).  Since the coordinate blocks and value intervals are ordered in the same way, global monotonicity follows once each block is kept internally sorted. The local update rule is the following. In a block \(B_j\), suppose the currently occupied entries, listed from left to right, are
\[
a_1<\cdots<a_s,
\qquad 0\le s\le s_j.
\]
If a new value \(x\in I_j\) arrives, the strategy replaces the leftmost occupied entry which is larger than \(x\), if such an entry exists.  If no occupied entry is larger than \(x\), then the strategy appends \(x\) in the next empty slot if one exists, and otherwise discards \(x\). This is exactly patience sorting, truncated at $s_j$ slots. See \cref{fig:replacement-patience} for an illustration. 

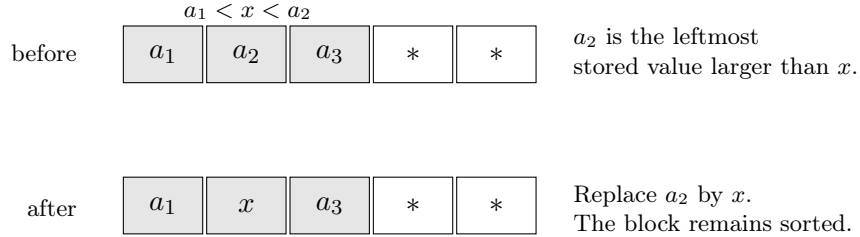
\begin{figure}[ht]
\centering
\begin{tikzpicture}[
  x=1.05cm,
  y=1cm,
  cell/.style={draw, minimum width=1.05cm, minimum height=.75cm, anchor=south west, inner sep=0pt},
  filled/.style={cell, fill=gray!20},
  empty/.style={cell},
  lab/.style={font=\footnotesize},
  arr/.style={-{Latex[length=2mm]}, thick}
]

\node[lab, anchor=east] at (0.55,4.90) {before};

\node[filled] at (1,4.50) {\(a_1\)};
\node[filled] at (2.05,4.50) {\(a_2\)};
\node[filled] at (3.10,4.50) {\(a_3\)};
\node[empty]  at (4.15,4.50) {\(\ast\)};
\node[empty]  at (5.20,4.50) {\(\ast\)};

\node[lab] at (2.575,5.45) {\(a_1<x<a_2\)};

\node[lab, anchor=east] at (0.55,2.85) {after};

\node[filled] at (1,2.5) {\(a_1\)};
\node[filled] at (2.05,2.5) {\(x\)};
\node[filled] at (3.10,2.5) {\(a_3\)};
\node[empty]  at (4.15,2.5) {\(\ast\)};
\node[empty]  at (5.20,2.5) {\(\ast\)};

\node[lab, align=left, anchor=west] at (6.55,4.90)
{
\(a_2\) is the leftmost\\
stored value larger than \(x\).
};

\node[lab, align=left, anchor=west] at (6.55,2.85)
{
Replace \(a_2\) by \(x\).\\
The block remains sorted.
};

\end{tikzpicture}

\caption{The local with-replacement rule inside a block.  If the incoming value \(x\) is smaller than some occupied entry, the strategy replaces the leftmost occupied entry larger than \(x\).  In the displayed case \(a_1<x<a_2\), so \(a_2\) is replaced by \(x\).  If no occupied entry is larger than \(x\), then \(x\) is appended in the next empty slot, if one exists.}
\label{fig:replacement-patience}
\end{figure}

\begin{lemma}\label{lem:replacement-admissible}
The strategy $S'_r$ is well defined and preserves global monotonicity of the array.
\end{lemma}

\begin{proof}
It suffices to check monotonicity inside each block.  Fix a block \(B_j\).  We prove by induction on the number of samples in $I_j$ that the occupied
entries of \(B_j\) form an increasing initial segment of the block.

The claim is trivial initially.  Suppose it holds before a sample \(x\in I_j\) is processed.  If some occupied entry is larger than \(x\), let \(a_i\) be the leftmost such entry.  Then \(a_{i-1}<x<a_i\), with the natural interpretation when \(i=1\).  Since \(a_i<a_{i+1}\) when \(i<s\), replacing \(a_i\) by \(x\) preserves the increasing order.  If no occupied entry is larger than \(x\), then \(x>a_s\), again with the obvious interpretation when \(s=0\).  Thus appending \(x\) in the next empty slot also preserves the increasing order.  Discarded samples do not change the state.

Hence each block remains internally sorted.  Since the intervals \(I_1,\dots,I_m\) and coordinate blocks \(B_1,\dots,B_m\) are ordered from left to right, internal monotonicity in each block implies global monotonicity of the full array.
\end{proof}

We next recall the standard invariant behind patience sorting. Given a sequence of real numbers \(x_1,\dots,x_t\), let \(\LIS(x_1,\dots,x_t)\) denote the length of the longest increasing subsequence.

\begin{lemma}\label{lem:truncated-patience}
Fix \(s\ge1\), and let \(x_1,\dots,x_t\) be distinct real numbers. Run the local rule above on a block with  capacity \(s\), starting from the empty configuration. Then the number of occupied slots after processing \(x_1,\dots,x_t\) is
\[
\min\{s,\LIS(x_1,\dots,x_t)\}.
\]
\end{lemma}

\begin{proof}
For \(1\le \ell\le s\), let \(a_\ell^{(u)}\) denote the \(\ell\)-th occupied value after processing \(x_1,\dots,x_u\), with the convention that \(a_\ell^{(u)}=+\infty\) if fewer than \(\ell\) slots are occupied. We prove by induction on \(u\) that \(a_\ell^{(u)}\) is the smallest possible last value of an increasing subsequence of length \(\ell\) in \(x_1,\dots,x_u\), with the convention that this value is \(+\infty\) if no such subsequence exists.

The claim is trivial for \(u=0\). Suppose it holds after \(u\) steps, and let \(x=x_{u+1}\). Let \(j\) be the smallest index such that
\[
a_j^{(u)}>x,
\]
setting \(j=s+1\) if no such index exists. The local rule replaces \(a_j^{(u)}\) by \(x\) if \(j\le s\), and makes no change if \(j=s+1\).

We check that this is exactly the update rule for the smallest possible last values of increasing subsequences. Fix \(1\le \ell\le s\). An increasing subsequence of length \(\ell\) in \(x_1,\dots,x_u,x\) either avoids \(x\), in which case its last value is at least \(a_\ell^{(u)}\), or uses \(x\) as its last value. The latter is possible if and only if there is an increasing subsequence of length \(\ell-1\) among \(x_1,\dots,x_u\) whose last value is smaller than \(x\). By the induction hypothesis, this is equivalent to
\[
a_{\ell-1}^{(u)}<x,
\]
where we set \(a_0^{(u)}=-\infty\).

Thus the smallest possible last value for a length-\(\ell\) increasing subsequence changes exactly when
\[
a_{\ell-1}^{(u)}<x<a_\ell^{(u)}.
\]
By the definition of \(j\), this happens precisely for \(\ell=j\), if \(j\le s\). In that case the new smallest possible last value is \(x\). For every other \(\ell\), the value remains unchanged. This agrees with the local update rule, so the invariant is preserved.

Finally, after processing \(x_1,\dots,x_t\), the occupied slots are precisely those \(\ell\le s\) for which \(a_\ell^{(t)}<+\infty\). By the invariant, these are precisely the values of \(\ell\le s\) for which there exists an increasing subsequence of length \(\ell\). Hence the number of occupied slots is
\[
\min\{s,\LIS(x_1,\dots,x_t)\}. \qedhere
\]
\end{proof}

\begin{algorithm}[t]
\caption{The with-replacement strategy \(S'_r\)}
\label{alg:with-replacement-block}
\DontPrintSemicolon
Partition the coordinates into \(B_1,\dots,B_m\) and \([0,1]\) into \(I_1,\dots,I_m\) as above.\;
\ForEach{\(j\in\{1,\dots,m\}\)}{
Set \(h_j\gets 0\).\;
}
\ForEach{new sample \(x\)}{
Let \(j\) be such that \(x\in I_j\), with endpoint ties broken arbitrarily.\;
Let \(a_1<\cdots<a_{h_j}\) be the currently occupied entries in \(B_j\), listed from left to right.\;
\If{there is an index \(i\le h_j\) with \(a_i>x\)}{
Let \(i\) be the smallest such index.\;
Replace \(a_i\) by \(x\).\;
}
\ElseIf{\(h_j<s_j\)}{
Place \(x\) in the next empty location of \(B_j\).\;
Set \(h_j\gets h_j+1\).\;
}
\Else{
Discard \(x\).\;
}
}
\end{algorithm}

\medskip

We proceed with the proof of \cref{thm:replacement-upper}. The proof is based on a one-block estimate and a restart argument.  We first record a crude counting bound for permutations with small longest increasing subsequence.

\begin{lemma}\label{lem:counting}
Let \(K\ge1\), \(s\ge2\), and let \(\pi\) be uniformly distributed on \(S_K\). Then
\[
\P(\LIS(\pi)<s)\le \frac{(s-1)^{2K}}{K!}.
\]
\end{lemma}

\begin{proof}
We bound the number of permutations with no increasing subsequence of length \(s\). Fix such a permutation \(\pi=\pi_1\cdots\pi_K\). For each position \(t\), let
\[
c(t):=\text{the length of the longest increasing subsequence of \(\pi\) ending at position \(t\)}.
\]
Then \(c(t)\in\{1,\dots,s-1\}\) for every \(t\).

If \(u<v\) and \(c(u)=c(v)\), then \(\pi_u>\pi_v\). Indeed, if \(\pi_u<\pi_v\), then an increasing subsequence ending at \(u\) could be extended by \(\pi_v\), contradicting \(c(u)=c(v)\). Hence, among positions with the same value of \(c(\cdot)\), the entries of \(\pi\) appear in decreasing order.

Now record the following two pieces of data:
\[
t\mapsto c(t),
\qquad
v\mapsto c(\pi^{-1}(v)).
\]
There are at most \((s-1)^K\) choices for the first assignment and at most \((s-1)^K\) choices for the second. Once these two assignments are fixed, there is at most one compatible permutation: for each \(q\in\{1,\dots,s-1\}\), the values assigned color \(q\) must be placed in the positions assigned color \(q\), in decreasing order.

Thus the number of permutations \(\pi\in S_K\) with \(\LIS(\pi)<s\) is at most \((s-1)^{2K}\). The probability bound follows by dividing by \(K!\).
\end{proof}

The next lemma translates the patience-sorting invariant into a one-block failure estimate.

\begin{lemma}\label{lem:one-block-discrete}
Fix a block \(B_j\), and suppose that it starts from the empty state. Let \(N_j(T)\) be the number of samples among \(X_1,\dots,X_T\) which fall in \(I_j\). Then, for every \(K\ge1\),
\[
\P(\text{\(B_j\) is unfinished after \(T\) samples})
\le
\P(N_j(T)<K)+\frac{(s_j-1)^{2K}}{K!}.
\]
\end{lemma}

\begin{proof}
If \(B_j\) is unfinished after \(T\) samples and \(N_j(T)\ge K\), then \(B_j\) was still unfinished after the first \(K\) samples falling in \(I_j\), since the number of occupied slots in a block is nondecreasing under the local rule.

The event \(\{N_j(T)\ge K\}\) depends only on which of the first \(T\) samples fall in \(I_j\). Conditioned on this event, the first \(K\) samples falling in \(I_j\) are independent and uniformly distributed on \(I_j\). Since these samples have a continuous distribution, their relative order is uniformly distributed on \(S_K\).

By \cref{lem:truncated-patience}, after these \(K\) samples the block is unfinished only if their longest increasing subsequence has length strictly less than \(s_j\). Applying \cref{lem:counting} gives the result.
\end{proof}

We next choose the length of one phase and bound the probability that a fixed block, started from the empty state, is not complete by the end of the phase.

\begin{lemma}\label{lem:one-phase-failure-discrete}
Let \(C\ge100\), and set
\[
K:=\lceil Cr^2\rceil,
\qquad
T_C:=\left\lceil \frac{2Kn}{r}\right\rceil.
\]
If a block \(B_j\) starts empty, then
\[
\P(\text{\(B_j\) is unfinished after \(T_C\) samples})
\le
2e^{-K/4}.
\]
\end{lemma}

\begin{proof}
Recall that \(s_j\in\{r,r+1\}\). Since \(|I_j|=s_j/n\),
\[
N_j(T_C)\sim \mathrm{Bin}\left(T_C,\frac{s_j}{n}\right).
\]
Its mean satisfies
\[
\mb{E}N_j(T_C)
=
T_C\frac{s_j}{n}
\ge
\frac{2Kn}{r}\cdot \frac{s_j}{n}
\ge
2K,
\]
because \(s_j\ge r\). Therefore, by the binomial Chernoff bound,
\[
\P(N_j(T_C)<K)
\le
\P\left(N_j(T_C)\le \frac12\mb{E}N_j(T_C)\right)
\le
e^{-\mb{E}N_j(T_C)/8}
\le
e^{-K/4}.
\]

It remains to bound the second term in \cref{lem:one-block-discrete}. Since \(s_j\le r+1\) and \(r\ge2\),
\[
s_j^2\le (r+1)^2\le \frac94r^2.
\]
As \(C\ge100\), this implies \(K\ge20s_j^2\). Using \(K!\ge (K/e)^K\), we get
\[
\frac{(s_j-1)^{2K}}{K!}
\le
\frac{s_j^{2K}}{K!}
\le
\left(\frac{es_j^2}{K}\right)^K
\le
\left(\frac e{20}\right)^K
\le
e^{-K}.
\]
Applying \cref{lem:one-block-discrete} with \(T=T_C\), we conclude that
\[
\P(\text{\(B_j\) is unfinished after \(T_C\) samples})
\le
e^{-K/4}+e^{-K}
\le
2e^{-K/4}. \qedhere
\]
\end{proof}

The final ingredient is that a partially completed block is no harder to finish than an empty one.

\begin{lemma}
\label{lem:fresh-phase}
Fix a block \(B_j\), and condition on any history of the process up to some time. For every \(T\ge0\), the conditional probability that \(B_j\) is unfinished after the next \(T\) samples is at most the probability that \(B_j\), started from the empty state, is unfinished after \(T\) samples.
\end{lemma}

\begin{proof}
If \(B_j\) is already complete, there is nothing to prove. Otherwise, fix its current state. Since this state is reachable, there is a sequence \(x_1,\dots,x_u\) of earlier samples falling in \(I_j\) which produces it when the block is started from the empty state.

Now fix the future samples falling in \(I_j\) during the next \(T\) global samples, and call them \(y_1,\dots,y_t\). Starting from the current state and processing \(y_1,\dots,y_t\) is the same as starting from the empty state and processing
\[
x_1,\dots,x_u,y_1,\dots,y_t.
\]
By \cref{lem:truncated-patience}, the two resulting occupancies are
\[
\min\{s_j,\LIS(x_1,\dots,x_u,y_1,\dots,y_t)\}
\qquad\text{and}\qquad
\min\{s_j,\LIS(y_1,\dots,y_t)\},
\]
respectively. The first quantity is at least the second, since every increasing subsequence of \(y_1,\dots,y_t\) is also an increasing subsequence of the concatenated sequence.

Thus, for every fixed realization of the future samples, the block can be unfinished from the current state only if it would also be unfinished from the empty state. Taking conditional probabilities proves the lemma.
\end{proof}

We now have all the ingredients needed to prove the with-replacement upper bound. 

\begin{proof}[Proof of \cref{thm:replacement-upper}]
Set
\[
r:=\max\left\{2,\left\lceil \sqrt{\log n}\right\rceil\right\},
\]
and run the with-replacement block strategy with this value of \(r\). Let \(\tau\) be its completion time.

We first prove the expectation bound. Take \(C=100\), and define
\[
K:=\lceil Cr^2\rceil,
\qquad
T_C:=\left\lceil \frac{2Kn}{r}\right\rceil.
\]
Let
\[
\delta_C:=2m e^{-K/4}.
\]
By \cref{lem:one-phase-failure-discrete} and the union bound, if all blocks start empty, then the probability that at least one block is unfinished after \(T_C\) samples is at most \(\delta_C\). Since \(m\le n/r\), \(K\ge100r^2\), and \(r^2\ge\log n\), we have
\[
\delta_C
\le
\frac{2n}{r}e^{-25r^2}
\le
\frac{2}{r}n^{-24}
<
\frac1{10}
\]
for all sufficiently large \(n\).

We now iterate phases of length \(T_C\). Let
\[
A_\ell:=\{\tau>\ell T_C\}.
\]
Condition on the history up to time \((\ell-1)T_C\). On \(A_{\ell-1}\), at least one block is unfinished. During the next \(T_C\) samples, the future samples are independent of the past. By \cref{lem:fresh-phase}, for each unfinished block the conditional probability of remaining unfinished at the end of this phase is at most the empty-start probability from \cref{lem:one-phase-failure-discrete}. A union bound over all blocks gives
\[
\P(A_\ell\mid \mc F_{(\ell-1)T_C})
\le
\delta_C
\qquad\text{on }A_{\ell-1}.
\]
Therefore
\[
\P(A_\ell)
=
\mb{E}\left[
\mbf1_{A_{\ell-1}}
\P(A_\ell\mid \mc F_{(\ell-1)T_C})
\right]
\le
\delta_C\P(A_{\ell-1})
\le
\delta_C^\ell.
\]
Using the tail-sum formula,
\[
\mb{E}[\tau]
=
\sum_{t\ge0}\P(\tau>t)
\le
\sum_{\ell\ge0}\sum_{t=\ell T_C}^{(\ell+1)T_C-1}\P(\tau>\ell T_C)
\le
T_C\sum_{\ell\ge0}\delta_C^\ell
=
\frac{T_C}{1-\delta_C}.
\]
Since \(\delta_C<1/10\) and \(K=\lceil100r^2\rceil\), this gives
\[
\mb{E}[\tau]\le \frac{10}{9}T_C=O(nr)=O(n\sqrt{\log n}).
\]

It remains to prove the high-probability statement. Fix \(D>0\). Choose \(C=C(D)\ge100\) large enough that
\[
\frac C4>D+2.
\]
Set
\[
K_D:=\lceil Cr^2\rceil,
\qquad
T_D:=\left\lceil \frac{2K_Dn}{r}\right\rceil.
\]
The strategy is unchanged; only the time at which we inspect it has changed. Since all blocks start empty at time \(0\), \cref{lem:one-phase-failure-discrete} and the union bound give
\[
\P(\tau>T_D)
\le
2m e^{-K_D/4}
\le
\frac{2n}{r}e^{-Cr^2/4}.
\]
As \(r^2\ge\log n\),
\[
\frac{2n}{r}e^{-Cr^2/4}
\le
\frac2r n^{1-C/4}
\le
n^{-D}
\]
for all sufficiently large \(n\). Finally,
\[
T_D=O_D(nr)=O_D(n\sqrt{\log n}).
\]
Thus, after increasing the constant and denoting it by \(C_D\),
\[
\P\left(\tau>C_D n\sqrt{\log n}\right)\le n^{-D}. \qedhere
\]
\end{proof}

\bibliographystyle{amsplain}
\bibliography{main}

@article{SamuelsSteele81,
  author = {Samuels, Stephen M. and Steele, J. Michael},
  title = {Optimal Sequential Selection of a Monotone Sequence from a Random Sample},
  journal = {The Annals of Probability},
  year = {1981},
  volume = {9},
  number = {6},
  pages = {937--947}
}

@article{Gnedin99,
  title={Sequential selection of an increasing subsequence from a sample of random size},
  author={Gnedin, Alexander V},
  journal={Journal of applied probability},
  volume={36},
  number={4},
  pages={1074--1085},
  year={1999},
  publisher={Cambridge University Press}
}

@article{BrussDelbaen01,
  author = {Bruss, F. Thomas and Delbaen, Freddy},
  title = {Optimal Rules for the Sequential Selection of Monotone Subsequences of Maximum Expected Length},
  journal = {Stochastic Processes and their Applications},
  year = {2001},
  volume = {96},
  number = {2},
  pages = {313--342}
}

@article{BrussDelbaen04,
  author = {Bruss, F. Thomas and Delbaen, Freddy},
  title = {A Central Limit Theorem for the Optimal Selection Process for Monotone Subsequences of Maximum Expected Length},
  journal = {Stochastic Processes and their Applications},
  year = {2004},
  volume = {114},
  number = {2},
  pages = {287--311}
}

@article{ArlottoNguyenSteele15,
  title={Optimal online selection of a monotone subsequence: a central limit theorem},
  author={Arlotto, Alessandro and Nguyen, Vinh V and Steele, J Michael},
  journal={Stochastic Processes and their Applications},
  volume={125},
  number={9},
  pages={3596--3622},
  year={2015},
  publisher={Elsevier}
}

@article{ArlottoMosselSteele16,
  title={Quickest online selection of an increasing subsequence of specified size},
  author={Arlotto, Alessandro and Mossel, Elchanan and Steele, J Michael},
  journal={Random Structures \& Algorithms},
  volume={49},
  number={2},
  pages={235--252},
  year={2016},
  publisher={Wiley Online Library}
}

@inproceedings{AamandAbrahamsenBerettaKleist23,
  author    = {Aamand, Anders and Abrahamsen, Mikkel and Beretta, Lorenzo and Kleist, Linda},
  title     = {Online Sorting and Translational Packing of Convex Polygons},
  booktitle = {Proceedings of the 2023 Annual ACM-SIAM Symposium on Discrete Algorithms (SODA)},
  pages     = {1806--1833},
  publisher = {SIAM},
  year      = {2023},
  doi       = {10.1137/1.9781611977554.ch69}
}

@article{AldousDiaconis99,
  author  = {Aldous, David and Diaconis, Persi},
  title   = {Longest increasing subsequences: from patience sorting to the {B}aik--{D}eift--{J}ohansson theorem},
  journal = {Bulletin of the American Mathematical Society},
  volume  = {36},
  number  = {4},
  pages   = {413--432},
  year    = {1999},
  doi     = {10.1090/S0273-0979-99-00796-X}
}

@inproceedings{AbrahamsenBerceaBerettaKlausenKozma24,
  author    = {Abrahamsen, Mikkel and Bercea, Ioana O. and Beretta, Lorenzo and Klausen, Jonas and Kozma, L{\'a}szl{\'o}},
  title     = {Online Sorting and Online {TSP}: Randomized, Stochastic, and High-Dimensional},
  booktitle = {32nd Annual European Symposium on Algorithms (ESA 2024)},
  series    = {Leibniz International Proceedings in Informatics (LIPIcs)},
  volume    = {308},
  pages     = {5:1--5:15},
  publisher = {Schloss Dagstuhl -- Leibniz-Zentrum f{\"u}r Informatik},
  year      = {2024},
  doi       = {10.4230/LIPIcs.ESA.2024.5}
}

@inproceedings{Hu26,
  author    = {Yang Hu},
  title     = {Nearly Optimal Bounds for Stochastic Online Sorting},
  booktitle = {Proceedings of the 2026 Annual ACM-SIAM Symposium on Discrete Algorithms (SODA)},
  pages      = {4969--4995},
  year       = {2026},
  publisher  = {Society for Industrial and Applied Mathematics (SIAM)},
  isbn       = {978-1-61197-897-1},
  doi        = {10.1137/1.9781611978971.180},
  url        = {https://doi.org/10.1137/1.9781611978971.180}
}

\end{document}